\documentclass[a4paper,USenglish]{llncs}

\usepackage{xcolor}
\usepackage{amsmath,amssymb,latexsym,stmaryrd}
\usepackage{hyperref}
\usepackage{listings}
\usepackage{paralist}

\usepackage{hhline}
\usepackage{verbatim}
\usepackage{datetime2}
\usepackage{xspace}
\usepackage[utf8]{inputenc}
\usepackage{tikz}
\usetikzlibrary{arrows,positioning,matrix,scopes,calc,decorations.pathmorphing}
\usepackage{cite}

\usepackage{cleveref}

\newcommand{\sharon}[1]{}
\newcommand{\sharonX}[1]{}
\newcommand{\SIX}[1]{}
\newcommand{\NRX}[1]{}
\newcommand{\OISHX}[1]{}
\newcommand{\noamKeepEmph}[1]{#1}

\newcommand{\colbf}[1]{\textrm{#1}}

\newcommand{\metaX}[1]{}

\newcommand{\ignore}[1]{}

\newcommand{\QUIC}{\textsc{Quic3}}

\newcommand{\NRparX}[1]{}

\newcommand{\ie}{\textit{i.e.}}

\newcommand{\rank}{rank}
\newcommand{\ranked}{ranked}
\newcommand{\shrinking}{squeezing}
\newcommand{\shrinkage}{squeezing}
\newcommand{\shrinker}{squeezer}
\newcommand{\Shrinker}{Squeezer}
\newcommand{\shrinkagefunc}{{\shrinkage} function}
\newcommand{\func}{\curlyvee}
\newcommand{\measure}{rank}
\newcommand{\m}{\rho}
\newcommand{\consecution}{simulation inducing}
\newcommand{\initiation}{initial anchor}
\newcommand{\safety}{fault preservation}

\newcommand{\recidivist}{recidivist}
\newcommand{\recidivism}{recidivism}

\newcommand{\Recidivism}{Recidivism}

\newcommand{\TS}{\textit{TS}}
\newcommand{\Init}{\textit{Init}}
\newcommand{\Tr}{\textit{Tr}}

\newcommand{\Spec}{\textit{P}}

\newcommand{\eqdef}{~\widehat{=}~}
\newcommand{\tighteq}{\hspace{1pt}{=}\hspace{1pt}}

\newcommand{\smallish}{\fontsize{6}{6}\selectfont}
\newcommand{\mediumish}{\fontsize{9}{9}\selectfont}

\newcommand{\setOfStates}{\mathrm{\Sigma}}

\newcommand{\alen}{\ell}
\newcommand{\Bound}{\alen_\textsf{B}}
\newcommand{\sq}{\func}
\newcommand{\sqa}[1]{\func\!(#1)}
\newcommand{\cstate}{\sigma}
\newcommand{\csqstate}[1][\cstate]{\sqa{#1}}

\newcommand{\prog}[1]{\lstinline!#1!}

\tikzset{
    table cell/.style={
        draw,
        inner xsep=0pt,
        outer sep=0pt
    },
    table/.style={
        matrix of nodes,
        row sep=-\pgflinewidth,
        column sep=-\pgflinewidth,
        inner sep=0pt,
        nodes={table cell}
    },
    math table cell/.style={
    	table cell,
        anchor=south
    },
    math table/.style={
        matrix of math nodes,
        row sep=-\pgflinewidth,
        column sep=-\pgflinewidth,
        nodes={math table cell}
    },
    minitable/.style={
    	math table, every node/.append style={
        	scale=0.6,
            minimum width=12mm, minimum height=5mm}
    }
}

\title{Putting the Squeeze on Array Programs:\\Loop Verification via Inductive Rank Reduction}

\titlerunning{Putting the Squeeze on Array Programs}
\author{
Oren Ish-Shalom\inst{1}
\and
Shachar Itzhaky\inst{2}
\and
Noam Rinetzky\inst{1}
\and
Sharon Shoham\inst{1}
}
\authorrunning{O. Ish-Shalom et al.}
\institute{Tel Aviv University, Israel \and Technion, Israel}

\begin{document}

\maketitle

\begin{abstract}
Automatic verification of array manipulating programs is a challenging problem because
it often amounts to the inference of inductive quantified loop invariants which, in some cases, may not even be first-order expressible.
In this paper, we suggest a novel verification technique that is based on induction on user-defined \emph{rank}  of program states as an alternative to loop-invariants.
Our technique, dubbed \emph{inductive rank reduction}, works in two steps. Firstly, we simplify the verification problem and prove that the program is correct when the input state contains an input array of length $\Bound$ or less, \noamKeepEmph{using the length of the array as the rank of the state.}{}
Secondly, we employ a \emph{squeezing function} $\sq$ which converts a program state $\cstate$ with an array of length $\alen > \Bound$ to a state $\csqstate$ containing an array of length $\alen-1$ or less. We prove that when $\sq$ satisfies certain natural conditions then if the program  violates its specification on $\cstate$  then it does so also on $\csqstate$. The correctness of  the program on inputs with arrays of arbitrary lengths  follows by induction.

We make our technique automatic for array programs whose length of execution is proportional to the length of the input arrays by (i) performing the first step using symbolic execution,
(ii) verifying the conditions required of $\sq$ using Z3,
and (iii) providing a heuristic procedure for synthesizing $\sq$.
We implemented our technique and applied it  successfully to several
interesting array-manipulating programs, including a bidirectional summation program whose loop invariant cannot be expressed in first-order logic while its specification is quantifier-free.

\end{abstract}

\pagestyle{plain}

\section{Introduction}\label{Se:Intro}

Automatic verification of array manipulating programs is a challenging problem because
it often amounts to the inference of inductive quantified loop invariants.
These invariants are frequently quite hard to come up with, even for seemingly
simple and innocuous program, both automatically and manually.
The purpose of this paper is to suggest an alternative kind of correctness
witness, which is often simpler \noamKeepEmph{than inductive invariants}{} and hence more amenable to automated search.

Loop invariants, the basis  of traditional verification approaches,
offer an induction scheme based on the time axis,
\textit{i.e.}, on the number of loop iterations.
We suggest an alternative approach in which induction is carried out
on the space axis, \text{i.e.} on a (user-defined notion of the) \emph{\rank} (e.g., size) of the program state.
This is particularly useful in the setting of infinite-state systems,
where the size of the state may be unbounded.
In this induction scheme,  establishing the induction step
relies on a \emph{{\shrinkagefunc}}
$\func : \setOfStates\to\setOfStates$ (read $\func$ as \emph{squeeze}) that maps
program states to lower-{\ranked} program states (up to a given minima).
Roughly speaking, the squeezing function should satisfy the following conditions, \noamKeepEmph{described here intuitively and formalized in \Cref{def:simulation}:}
\begin{itemize}
\item \textbf{Initial anchor.}~ $\sq$ maps initial states to initial states.
\item \textbf{Simulation inducing.}~ $\sq$ induces a certain form of simulation between the program states and their squeezed counterparts.
\item \textbf{Fault preservation.}~ $\sq$ maps unsafe states to unsafe states.
\end{itemize}

Our main theorem (\Cref{thm:soundness}) shows that if these conditions are satisfied then  $P$ is correct, provided it is correct on its \emph{base}, i.e., on the states with minimal \rank.
The crux of the proof is that as a consequence of the aforementioned conditions, if $P$  violates its specification on a state $\cstate$  then it also violates it  on~$\csqstate$.  Hence, if $P$ satisfies the specification on the base states, by induction it satisfies it on any state.

The function $\sq$  itself can be given by the user or, as we show in \Cref{Se:Syn}, automatically obtained for a  class  of array programs which iterate over their input arrays looking for a particular element (e.g., \lstinline!strchr!) or aggregating their elements (e.g., \lstinline!max!). 
In our experiments, we utilized automatically synthesized squeezing functions to  verify   natural specifications of  several interesting array-manipulating programs, some of which are beyond the capabilities of existing automatic techniques.
Arguably, the key benefit of the our approach is that the {\shrinking} functions are often rather simple, and thus finding them and establishing that they satisfy the required properties is an easier task than the inference of loop invariants. For example, in the next section we show a program whose loop invariant cannot be expressed in first order logic but can be proven correct using a \shrinking{} function which is first-order expressible,  in fact, the reasoning about the automatically synthesized \shrinking{} function is quantifier free.

The last point to discuss is the verification of the program on states in the base of $\sq$. Here, we apply standard verification techniques but to a simpler problem: we need to establish correctness only on the base, a rather small subset of the entire state space. For example, for the programs in our experiments it is possible to utilize symbolic execution to verify the correctness of the programs on all arrays of length three or less. This approach is effective because on the programs in our benchmarks, the bound on the length of the input arrays also determines a bound on the length of the execution.
As this aspect of our technique is rather standard we do not discuss it any further.

\paragraph{Outline.}
The rest of the paper is structured as follows:
We first give an informal overview of our approach (\Cref{Se:Overview}) which is followed by a formal definition of our technique and a proof of its soundness (\Cref{Se:shrinker}).
We continue with a description of our heuristic procedure for  synthesizing squeezing functions (\Cref{Se:Syn})
and a discussion about our implementation and experimental results (\Cref{Se:ImpExper}).
We then review closely related work (\Cref{Se:Related}) and conclude (\Cref{Se:Conc}).

\section{Overview}\label{Se:Overview}

In this section, we give a high-level view of our technique. 

\paragraph{Running example.}  Program \lstinline|sum_bidi|, shown in \Cref{intro:sum-bidi-loopinv},
computes the sum of the input array \lstinline!a! in two ways:
One computation accumulates elements from left to right, and the other ---
from right to left (assuming that indexes grow to the right).
Ignoring its dubious usefulness, \lstinline!sum_bidi! possesses an intricate
property: the variables \lstinline!l! and \lstinline!r! are both computed
to be the sum of the input array \lstinline!a!.
A natural property one expect to hold when the program terminates is that $l=r$. 

\paragraph{The challenge.}
To verify the aforementioned postcondition 
when the length of the array is not known and \emph{unbounded},
a loop invariant is \noamKeepEmph{often employed}. 
It is important to remember, that a loop invariant must hold on all
intermediate loop states --- every time execution hits the loop header. 
For this reason, the loop invariant needed in this case is more involved than the mere assertion $l = r$ that follows the loop. 
The right side of \Cref{intro:sum-bidi-loopinv} shows a possible
loop invariant for this scenario.
Intuitively, the invariant says that  \lstinline!l! and \lstinline!r! differ by the sum of the elements that they have not yet, respectively, accumulated.
Notice that the invariant's formulation relies on a function $\mathrm{sum}(\cdot)$
for arrays (and array slices), the definition of which is also included
in the figure.
This definition is recursive;
indeed, any definition of sum will require some form
of recursion or loop due to the unbounded sizes of arrays in program
memory.
This kind of ``logical escalation'' (from quantifier-free $l=r$ to a
fixed-point logic) makes such verification tasks challenging, since
modern solvers are not particularly effective in the presence of quantifiers
and recursive definitions.

Moreover, a system attempting to automate discovery of such loop invariants
is prone to serious scalability issues since it has to discover the
definition of $\mathrm{sum}(\cdot)$ along the way.
The subject program \lstinline|sum_bidi| effectively computes a sum,
so this auxiliary definition is at the same scale of complexity as
the program itself.

\vspace{1em}

\begin{figure}[t]
\vspace{-2mm}
\begin{tabular}{l@{\hspace{4mm}}|@{\hspace{4mm}}l}
\mediumish
\begin{lstlisting}
void sum_bidi(int a[], int n) {
   int l = 0, r = 0;
   for (int i = 0; i < n; i++) {
      l += a[i];
      r += a[(*@n - i - 1@*)];
   }
   (*@\color{red}assert@*)(l == r);
}
\end{lstlisting}
&
$
\begin{array}{l}
\begin{array}{l@{}l}
  I\eqdef \big( & l + \mathrm{sum}(a[i:n]) = \\[0.25em]
                & r + \mathrm{sum}(a[0:n-i]) \big)
\end{array} \\[1.5em]
\begin{array}{l}
  \mathrm{sum}(a[j:k]) \eqdef{}\\
  \quad \mathbf{if}~j<k \\
  \qquad \mathbf{then}~ a[j] + sum(a[j+1:k]) \\
  \qquad \mathbf{else}~0
\end{array}
\end{array}$
\vspace{-2mm}
\end{tabular}
\caption{A bidirectional sum example and a loop invariant for it.}
\label{intro:sum-bidi-loopinv}
\end{figure}

\paragraph{Our approach.}
We suggest to leverage the semantics already present in the subject program
for a more compact proof of safety.
Instead of having to summarize partial executions of the program via a loop invariant, we show that the program is correct for
all arrays of size $0...r$ for some \emph{base rank} $r$
(the size of the array serves as the {\rank} of the program state),
and further show how to derive the correctness of the program for arrays of size $n > r$, from its correctness for arrays of size $n - 1$.
To achieve the latter, we rely on a function that ``squeezes'' states in which the array length is $n$ to states in which the array length is $n - 1$, as we illustrate next.

Continuing with the example \lstinline|sum_bidi| described above,
we use the function $\func:\setOfStates\to\setOfStates$,
defined as a code block on the right side of \Cref{intro:sum-bidi},
to ``squeeze'' program states.
In this case, the state consists
of the variables $\langle$\lstinline|a|, \lstinline|n|, \lstinline|i|,
\lstinline|l|, \lstinline|r|$\rangle$, and it is squeezed by removing the first element of
\lstinline|a| and adjusting the indices and sums accordingly.
The base rank here is $r=0$, since any non-empty array can be
squeezed in this manner.
The bottom part of \Cref{intro:sample-trace} shows the effect of applying
$\func$ to each of the states in the execution trace of \lstinline|sum_bidi|
on the example input \lstinline|[7,2,9,1,4]|.
The first property that is demonstrated by the diagram is the ``initial anchor'' property, stating that initial states are ``squeezed'' into initial states.
As is obvious from the diagram, the execution on the squeezed array
\lstinline|[2,9,1,4]| is accordingly shorter, so $\func$ cannot be
injective --- in this case, $\func(\sigma_0)=\func(\sigma_1)=\sigma'_0$.
Still, the sequence
$\sigma'_0\to \sigma'_1\to\sigma'_2\to\sigma'_3\to\sigma'_4$
constitutes a valid trace of \lstinline|sum_bidi|.
This is the second property required of $\func$, which we refer to as \emph{simulation inducing} and define it
formally in the next section.

\begin{figure}[t]
\vspace{-2mm}
\begin{tabular}{l@{\hspace{4mm}}|@{\hspace{4mm}}p{5cm}}
\begin{tikzpicture}[>=stealth,baseline=(o)]
  \def\sza{8mm}
  \matrix (m) at (0,1)
    [table, every node/.append style={minimum width=\sza,
                               minimum height=\sza}] {
     |(a0)| 7 & |(a1)| 2 & |(a2)| 9 & |(a3)| 1 & |(a4)| 4 \\
  };
  
  \node[left=0 of m](s) {$\sigma=$};
  
  \node[above left=0 of s] {\smallish$\mathrm{rank}(\sigma)=5$};
  
  \node[above=0 of a0] {\smallish \lstinline|l|$\tighteq 9$};
  \node[above=0 of a4] {\smallish \lstinline|r|$\tighteq 5$};
  \node[above=1mm of a2,anchor=south west](i) 
  	{\smallish \lstinline|i|$\tighteq 2$};
  \draw[->] (i.west) -- (a2);
  
  \coordinate(o) at (0,0);
  
  \matrix (m') at (0,-0.8)
    [table, every node/.append style={minimum width=\sza,
                               minimum height=\sza}] {
     |(a'0)| 2 & |(a'1)| 9 & |(a'2)| 1 & |(a'3)| 4 \\
  };
  
  \node[left=0 of m'](s') {$\func(\sigma)=$};

  \node[below left=1mm of s', inner sep=0, xshift=4mm] {\smallish$\mathrm{rank}(\func(\sigma))=4$};

  \node[above=0 of a'0] {\smallish \lstinline|l|$\tighteq 2$};
  \node[above=0 of a'3] {\smallish \lstinline|r|$\tighteq 4$};
  \node[above=1mm of a'1,anchor=south west](i')
  	{\smallish \lstinline|i|$\tighteq 1$};
  \draw[->] (i'.west) -- (a'1);

  \draw (s.-150) edge[->,dashed,out=-120,in=120] 
  	node[left] {$\func$} (s'.150);

\end{tikzpicture}
&
\mediumish
\vspace{-2cm}
\begin{lstlisting}
(*@$\func:$@*) {
      if (i > 0) {
         remove(a,0);
         i--;
         l -= a[0];
         r -= a[(*@n - i@*)];
      } else {
         remove(a,0);
      }
   }
\end{lstlisting}
\vspace{-1cm}
\end{tabular}
\vspace{-3mm}
\caption{A bidirectional sum example and its {\shrinkagefunc}.}
\label{intro:sum-bidi}
\end{figure}

Now, draw attention to \emph{fault preservation}, the third property required of $\func$:
whenever a state $\sigma$ falsifies the safety property $\varphi$, denoted $\sigma\not\models \varphi$,
it is also the case that $\func(\sigma)$ falsifies the safety property, i.e. $\func(\sigma) \not\models \varphi$. In our example,
the safety property can be formalized as $\varphi \eqdef (i = n \to l = r)$.
The reasoning establishing fault preservation is not immediate but still quite simple: if
$\sigma\not\models\varphi$, it means that $i=n$ but $l\neq r$ (at $\sigma$).
In that case, $a[n-i]=a[0]$; so $l'=l-a[0]\neq r-a[n-i]=r'$, where
$l'$, $r'$ are the values of \lstinline|l| and \lstinline|r|,
respectively, at state $\func(\sigma)$.
Since $i$ and $n$ are both decremented%
\footnote{Notice that we assume a positive size ($n>0$), otherwise the array
  cannot be squeezed in the first place.}
we get $\func(\sigma)\not\models\varphi$.

\begin{figure}[t]
\vspace{-2mm}
\centering
\begin{tikzpicture}[>=latex]
  \def\sza{8mm}
  \matrix (m1) at (0, -0.2)
    [table, every node/.append style={minimum width=\sza,
                               minimum height=\sza}] {
     |(a0)| 7 & |(a1)| 2 & |(a2)| 9 & |(a3)| 1 & |(a4)| 4 \\
  };
  \node[left=2pt of m1.north west,anchor=north east,scale=0.9] {$a=$};
  \matrix at (0,-2) [matrix of nodes,column sep=12mm] {
  	|(s0)| $\sigma_0$ & |(s1)| $\sigma_1$ &
  	|(s2)| $\sigma_2$ & |(s3)| $\sigma_3$ &
  	|(s4)| $\sigma_4$ & |(s5)| $\sigma_5$ \\
  };
  \draw[->] (s0) -- node[below] {\tiny TR} (s1);
  \draw[->] (s1) -- node[below] {\tiny TR} (s2);
  \draw[->] (s2) -- node[below] {\tiny TR} (s3);
  \draw[->] (s3) -- node[below] {\tiny TR} (s4);
  \draw[->] (s4) -- node[below] {\tiny TR} (s5);
  
  \matrix (s0m) [minitable, above=0 of s0] {
    a \tighteq a & |(s0i)| i \tighteq 0 \\
    l \tighteq 0 &         r \tighteq 0 \\
  };
  \matrix (s1m) [minitable, above=0 of s1] {
    a \tighteq a & |(s1i)| i \tighteq 1 \\
    l \tighteq 7 &         r \tighteq 4 \\
  };
  \matrix (s2m) [minitable, above=0 of s2] {
    a \tighteq a & |(s2i)| i \tighteq 2 \\
    l \tighteq 9 &         r \tighteq 5 \\
  };
  \matrix (s3m) [minitable, above=0 of s3] {
    a \tighteq a  & |(s3i)| i \tighteq 3 \\
    l \tighteq 18 &         r \tighteq 14 \\
  };
  \matrix (s4m) [minitable, above=0 of s4] {
    a \tighteq a  & |(s4i)| i \tighteq 4 \\
    l \tighteq 19 & r \tighteq 16 \\
  };
  \matrix (s5m) [minitable, above=0 of s5] {
    a \tighteq a  & |(s5i)| i \tighteq 5 \\
    l \tighteq 23 &         r \tighteq 23 \\
  };
  
  \node[right=0 of s5m] {$\models\varphi$};
  
  \node(terminal)[right=2mm of m1] {};
  
  { [every node/.style={circle,fill,inner sep=1pt}]
  \node at (a0.south) {};
  \node at (a1.south) {};
  \node at (a2.south) {};
  \node at (a3.south) {};
  \node at (a4.south) {};
  \node at (a4.south -| terminal) {};
  }
  
  \draw[dotted] (s0i.45) -- (a0.south);
  \draw[dotted] (s1i) -- (a1.south);
  \draw[dotted] (s2i) -- (a2.south);
  \draw[dotted] (s3i) -- (a3.south);
  \draw[dotted] (s4i) -- (a4.south);
  \draw[dotted] (s5i.135) -- (a4.south -| terminal);
  
  \matrix at (0,-3.1) [matrix of nodes,column sep=12mm] {
  	|(s'0)| $\sigma'_0$ & |(s'1)| $\sigma'_1$ &
  	|(s'2)| $\sigma'_2$ & |(s'3)| $\sigma'_3$ &
  	|(s'4)| $\sigma'_4$ \\
  };
  \draw[->] (s0) -- node[left,xshift=-1pt] {\small$\func$} (s'0);
  \draw[->] (s1) -- node[left] {\small$\func$} (s'0);
  \draw[->] (s2) -- node[left] {\small$\func$} (s'1);
  \draw[->] (s3) -- node[left] {\small$\func$} (s'2);
  \draw[->] (s4) -- node[left] {\small$\func$} (s'3);
  \draw[->] (s5) -- node[left] {\small$\func$} (s'4);
  
  \draw[->] (s'0) -- node[above] {\tiny TR} (s'1);
  \draw[->] (s'1) -- node[above] {\tiny TR} (s'2);
  \draw[->] (s'2) -- node[above] {\tiny TR} (s'3);
  \draw[->] (s'3) -- node[above] {\tiny TR} (s'4);
    
  \matrix (s'0m) [minitable, below=0 of s'0] {
    a \tighteq a' & |(s'0i)| i \tighteq 0 \\
    l \tighteq 0  & r \tighteq 0 \\
  };
  \matrix (s'1m) [minitable, below=0 of s'1] {
    a \tighteq a' & |(s'1i)| i \tighteq 1 \\
    l \tighteq 2  & r \tighteq 4 \\
  };
  \matrix (s'2m) [minitable, below=0 of s'2] {
    a \tighteq a' & |(s'2i)| i \tighteq 2 \\
    l \tighteq 11 & r \tighteq 5 \\
  };
  \matrix (s'3m) [minitable, below=0 of s'3] {
    a \tighteq a'  & |(s'3i)| i \tighteq 3 \\
    l \tighteq 12  &          r \tighteq 14 \\
  };
  \matrix (s'4m) [minitable, below=0 of s'4] {
    a \tighteq a'  & |(s'4i)| i \tighteq 4 \\
    l \tighteq 16  &          r \tighteq 16 \\
  };
  
  \node[right=0 of s'4m] {$\models\varphi$};
  
  \matrix (m2) at (4mm, -5)
    [table, every node/.append style={minimum width=\sza,
                               minimum height=\sza}] {
    |(a'0)| 2 & |(a'1)| 9 & |(a'2)| 1 & |(a'3)| 4 \\
  };  
 
  \node[left=2pt of m2.south west,anchor=south east,scale=0.9] {$a'=$};

  { [every node/.style={circle,fill,inner sep=1pt}]
  \node at (a'0.north) {};
  \node at (a'1.north) {};
  \node at (a'2.north) {};
  \node at (a'3.north) {};
  \node at (a'3.north -| terminal) {};
  }

  \draw[dotted] (s'0i) -- (a'0.north);
  \draw[dotted] (s'1i) -- (a'1.north);
  \draw[dotted] (s'2i.-30) -- (a'2.north);
  \draw[dotted] (s'3i) -- (a'3.north);
  \draw[dotted] (s'4i) -- (a'3.north -| terminal);
   
\end{tikzpicture}
\vspace{-2mm}
\caption{Example trace of \lstinline|sum_bidi|, and the corresponding
  shrunken image.}
\label{intro:sample-trace}
\end{figure}
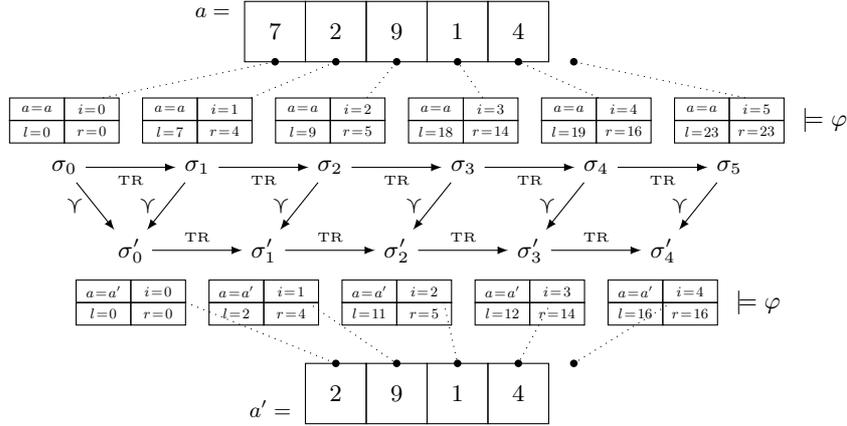

In this manner, from \emph{the assumption that $\func(\sigma_j)$, for $j=0..5$, induces a safe
trace}, we conclude that $\sigma_j$ is safe as well.
This lends the notion of constructing a proof by induction on the size of
the initial state $\sigma_0$, provided that $\func$ cannot ``squeeze forever''
and that we can verify all the minimal cases more easily, \textit{e.g.}
with bounded verification.
This is definitely true for \lstinline|sum_bidi|, since the minimal case
would be an empty array, in which the loop is never entered.
In some situations the minima contains states with small but not empty arrays.
In general, if one can verify that the program is correct when started with a minimal initial state, thus
establishing the base case of the induction, our technique would lift this proof to hold for unbounded initial states.
In particular, if the length of the program's execution trace can be bounded based on the size of the initial state then bounded model checking and symbolic execution can be lifted to obtain
unbounded correctness guarantee.

It is worth mentioning at this point that $\func$ is in no sense ``aware'' that it is,
in fact, reasoning about sums.
It only has to handle scalar operations, in this case subtraction (as the counterpart of addition
that occurs in \lstinline|sum_bidi|;
the same will be true for any other commutative, invertible operation.)
The folding semantics arises spontaneously from the induction over the size of the array.


\paragraph{Recap.}
We suggest a novel verification technique that is based on induction on the size of the input states as an alternative to loop-invariants.
The technique is based on utilizing a {\shrinking} function which converts high-\ranked{} 
states into low-\ranked{} 
 ones, and then  applying a standard verification technique to establish the correctness of the program on the \noamKeepEmph{minimally-ranked}{} 
states.
In a manner analogous to that which is carried out with ``normal''
verification using loop invariants, the {\shrinker} has to uphold the three
properties described in \Cref{Se:Intro}, namely \emph{initial anchor}, \emph{simulation inducing}, and \emph{fault preservation}. (See \Cref{Se:shrinker} for a formal definition.)
These properties ensure that the mapping induces a valid reduction between the safety
of any trace and that of its squeezed counterpart.

\medskip

\begin{paragraph}{Why bother.}
  The attentive readers may ask themselves, given that both loop
  invariants and squeezers incur some proof obligations for them to be
  employed for verification,
  what benefit may come of favoring the latter over the former.
  While the verification condition scheme proposed here is not inherently
  simpler (and arguably less so) than its Floyd-Hoare counterpart,
  we would like to point out that the \emph{{\shrinker} itself}, at least
  in the case of \lstinline|sum_bidi|, \emph{is indeed simpler} than
  the loop invariant that was needed to verify the same specification.
  It is simpler in a sense that it resides in a \emph{weaker logical
  fragment}:
  while the invariant relies on having a definition of (partial) sums,
  itself a recursive definition,
  the {\shrinker} $\func$ can be axiomatized in a quantified-free
  formula using a theory of strings \noamKeepEmph{(sequences)~\cite{Seq}}{} and
  linear arithmetic.
  In \Cref{Se:Syn} we take advantage of the simplicity if the {\shrinking} function,
  and show that it is feasible to \emph{generate it automatically} using
  a simple enumerative synthesis procedure.
\end{paragraph}

On top of that, it is quite immediate to see that the induction scheme
outlined above is still sound even if the properties of $\func$
(initial anchor, simulation, and fault preservation) only hold for
\emph{reachable} states.
Obviously, the set of reachable states cannot be expressed directly ---
otherwise we would have just used its axiomatization together with the
desired safety property, making any use of induction superfluous.
Even so, if we can acquire any known property of reachable states,
\textit{e.g.} through a preliminary phase of abstract interpretation~\cite{CousotCousot77},
then this property can be added as an assumption, simplifying $\func$
itself.
A keen reader may have noticed that the specification of
\lstinline|sum_bidi| has been written down as $\varphi \eqdef (i=n\to l=r)$,
while a completely honest translation of the assertion would in fact
produce a slightly stronger form, $\varphi' \eqdef (i\geq n \to l=r)$.
This was done for presentation purposes; in an actual scenario the
``proper'' specification $\varphi'$ is used, and a premise
$0\leq i \leq n$ is assumed.
Such range properties are prevalent in programs with arrays and indexes,
and can be discovered easily using static analysis, e.g., using the Octagon domain~\cite{Octagon}.

This final point is encouraging because it gives rise to a hybrid
approach, where a \emph{partial} loop invariant is used as a baseline ---
verified via standard techniques ---
and is then \emph{stengthened} to the desired safety property via
{\shrinker}-based verification.
Or, the order could be reversed.
There can even be alternating strengthening phases each using a different
method.
These extended scenarios are potentialities only and are matter for future work.

\section{Verification by Induction over State Size}
\label{Se:shrinker}

In this section we formalize our approach for verifying programs that operate over states (inputs) with an unbounded size. The approach mimics induction over the state size.
The base case of the induction is discharged by verifying the program for executions over ``small'' low-\ranked{} states (to be formalized later).
For the induction step, we need to deduce correctness of executions over ``larger'' higher-\ranked{} states from the correctness of executions over ``smaller'' states.
This is facilitated by the use of a \emph{simulation-inducing {\shrinkagefunc}} $\func$. Intuitively, the function transforms a state $\cstate$ into a corresponding ``smaller'' state $\func(\cstate)$ such that executions starting from the latter simulate executions starting from the former. The simulation ensures that correctness of the executions starting from
the smaller state,
$\func(\cstate)$, implies correctness of the executions starting from the larger one,
$\cstate$.

\paragraph{Transition systems and safety properties.}
To formalize our technique, we first define the semantics of programs using \emph{transition systems}.
The is quite standard.

\begin{definition}[Transition Systems]
\label{def:ts}
A \emph{transition system} $\TS = (\setOfStates, \Init, \Tr, \Spec)$ is a quadruple comprised of a
\emph{universe} (a set of states) $\setOfStates$, a set of \emph{initial states} $\Init \subseteq \setOfStates$, a \emph{transition relation} $\Tr \subseteq \setOfStates \times \setOfStates$, and a set of \emph{good  states} $\Spec \subseteq \setOfStates$.
\end{definition}
A \emph{trace} of $\TS$ is a (finite or infinite) sequence of states $\tau = \sigma_0,\sigma_1,\ldots$ such that for every $0 \leq i < |\tau|$, $(\sigma_i,\sigma_{i+1}) \in \Tr$.
In the following, we write $\Tr^k$, for $k \geq 0$ to denote $k$ self compositions of $\Tr$, where $\Tr^0=\mathit{Id}$ denotes the identity relation.
That is, $(\sigma,\sigma') \in \Tr^k$ if and only if $\sigma'$ is reachable from $\sigma$ by a trace of length $k$
(where the length of a trace is defined to be the number of transitions along the trace).

A transition system $\TS = (\setOfStates, \Init, \Tr, \Spec)$ is \emph{safe} if all its \emph{reachable states} are good (or ``safe''), where the set of reachable states is defined, as usual, to be the set of all states that reside on traces that start from the initial states. A \emph{counterexample trace} is a trace that starts from an initial state and includes a ``bad'' state, i.e., a state that is not in $\Spec$. The transition system is safe if and only if it has no counterexample traces.

\paragraph{Simulation-inducing squeezer.}
To present our technique, we start by formalizing the notion of a simulation-inducing {\shrinkagefunc} (\emph{\shrinker} for short).

\begin{definition}[Squeezing function]
\label{definition_squeezing_function}
Let $X$ be a set and $\preceq$ a well-founded partial order over $X$.
Let $B \supseteq \min(X)$ be a \emph{base} for $X$, where $\min(X)$ is the set of all the minimal elements of $X$ w.r.t.\ $\preceq$, and
let $\m: \setOfStates \to X$ be a \emph{\measure} on the program states.
A function $\func: \setOfStates \to \setOfStates$ is a \emph{{\shrinkagefunc}}, or \emph{\shrinker} for short, \emph{with base $B$} if for every state $\sigma \in \setOfStates$ such that $\m(\sigma) \in X \setminus B$, it holds that $\m(\func(\sigma)) \prec  \m(\sigma)$.
\end{definition}
That is, $\func$ must strictly decrease the {\rank} of any state unless its {\rank} is in the base, $B$.
We refer to states whose size is in $B$ as \emph{base states}, and denote them $\setOfStates_B = \{\sigma \in \setOfStates \mid \m(\sigma) \in B\}$. We denote by $\setOfStates_{\overline{B}} = \setOfStates \setminus \setOfStates_B$ the remaining states.
Since $\preceq$ is well-founded and all the minimal elements of $X$ w.r.t.\ $\preceq$ must be in $B$ (additional elements may be included as well),
any maximal strictly decreasing sequence of elements from $X$ will reach $B$ (i.e., will include at least one element from $B$).
Hence, the requirement of a {\shrinker} ensures that any state will be transformed into a base state by a \emph{finite} number of $\func$ applications.

\begin{example}\label{Ex:Squeeze}
In our examples, we use $(\mathbb{N}, \leq)$ as a well-founded set, 
and define the base as an interval $[0,k]$ for some (small) $k\geq 0$.
While it suffices to define $B =\min(\mathbb{N})=\{0\}$, it is sometimes beneficial to extend the base to an interval since it excludes additional states from the squeezing requirement of $\func$ (see \Cref{Se:ImpExper}).
For array-manipulating programs, the \measure{} used is often (but not necessarily) the size of the underlying array, in which case, the ``squeezing'' requirement is that whenever the array size is greater than $k$, the {\shrinker} must remove at least one element from the array. For example, for \lstinline|sum_bidi| (\autoref{intro:sum-bidi}), we consider $k=0$, i.e., the base consists of arrays of size $0$, and, indeed, whenever the array size is greater than $0$, it is decremented by $\func$. For arrays of size $0$, $\func$ behaves as the identity function (this case is omitted from the figure).
In addition, whenever the state contains more than one array,
we will use the sum of lengts of all arrays as a rank.
\end{example}

\begin{definition}[Simulation-inducing {\shrinker}]
\label{def:simulation}
Given a transition system $\TS = (\setOfStates, \Init, \Tr, \Spec)$, a {\shrinker} $\func: \setOfStates \to \setOfStates$ 
is \emph{simulation-inducing} if the following three conditions hold for every $\sigma \in \setOfStates$: 
\begin{itemize}[$\bullet$]
\item {\bf Initial anchor:} if $\sigma \in \Init$ then $\func(\sigma) \in \Init$ as well.
\item {\bf Simulation inducing:} there exist $n_\sigma \geq 1$ and $m_\sigma \geq 0$ such that if $(\sigma,\sigma')\in \Tr^{n_\sigma}$ then $(\func(\sigma),\func(\sigma'))\in \Tr^{m_{\sigma}}$, i.e., if $\sigma$ reaches $\sigma'$ in $n_\sigma$ steps, then the same holds for their $\func$-images, except that the number of steps may be different.
\item {\bf Fault preservation:} if $\sigma \not\in \Spec$ then $\func(\sigma) \not\in \Spec$ as well.
\end{itemize}
\end{definition}
The definition implies that $\{(\sigma,\func(\sigma)) \mid \sigma \in \setOfStates \}$ 
is a form of a ``skipping'' simulation relation, where steps taken  both from the simulated state, $\sigma$, and from the simulating state, $\func(\sigma)$, may skip over some states. This allows the simulated and the simulating execution to proceed in a different pace, but still remain synchronized. 
In fact, to ensure that we obtain a ``skipping'' simulation, it suffices to consider a weaker
simulation inducing 
requirement where the parameter $m_\sigma$ that determines the number of steps in the simulating trace depends not only on $\sigma$ but also on $\sigma'$ and may be different for each $\sigma'$. Note that for deterministic programs (as we use in our experiments) these   requirements are equivalent.
Another possible, yet stronger, relaxation is to weaken the requirement that $(\func(\sigma),\func(\sigma'))\in \Tr^{m_{\sigma}}$ into $(\func(\sigma),\func(\sigma'))\in \Tr^{i}$ for some $0 \leq i \leq m_{\sigma}$.

\begin{example}
To illustrate the \consecution{} requirement, recall the program \lstinline|sum_bidi| from \Cref{Ex:Squeeze}. For the base states ($n=0$), $\func$ behaves as the identity function. 
Hence, for such states the skipping parameters $n_\sigma$ and $m_\sigma$ are both $1$ (letting each step be simulated by itself). For non-base states, $n_\sigma$, the ``skipping'' parameter of $\sigma$,  is still $1$, while $m_\sigma$, the ``skipping'' parameter of $\func(\sigma)$, is $0$ if $\sigma$ is an initial state, and $1$ otherwise. This accounts for the fact that $\func$ truncates the head of the array; hence, the first step in an execution is skipped in the corresponding ``squeezed'' execution, while the rest of the steps are synchronized in both executions (see \Cref{intro:sample-trace} for an illustration).
\end{example}

Intuitively, one may conjecture that given a loop that iterates over an
array, it will essentially perform fewer iterations when run on
$\func(\sigma)$ than it does on $\sigma$, always resulting in
$m_\sigma \leq n_\sigma$. The following example shows that this is not
necessarily the case.

\begin{figure}
\begin{tabular}{@{}p{6.5cm}@{~~}|@{~~}l@{}}
\vspace{-1cm}
\begin{lstlisting}
bool is_sorted(int a[], int n) {
  for (int i = 1; i < n; i++)
    if (a[i] < a[i-1])
      return false;
  return true;
}
\end{lstlisting}
\vspace{-1cm}
&
\begin{lstlisting}
(*@$\func$@*): if (a[n-3] <= a[n-2] &&
       a[n-2] <= a[n-1])
        remove(a,n-1);
   else remove(a,n-4);
\end{lstlisting}
\end{tabular}
\caption{Another program with $\func$ demonstrating a
  scenario where $n_\sigma<m_\sigma$.}
\label{shrinker:is_sorted}
\end{figure}

\begin{example}
The program \lstinline|is_sorted| (\Cref{shrinker:is_sorted})
checks whether the input array elements are ascending by
comparing all consecutive pairs.
Our squeezer (for $n > 3$) checks whether the last three elements form
an ascending sequence; if so, removes the last element, otherwise
it removes the forth element from the right.
Consider the input a={1,0,2,3,1} and the squeezed a'={1,2,3,1}.
\lstinline|is_sorted(a)| terminates after one iteration, but
\lstinline|is_sorted(a')| after three iterations.
Let $\sigma=\big[a,i\mapsto 1\big]$. The simulation inducing requirement
can only be satisfied with $n_\sigma=1$ and $m_\sigma=3$.
Since $\Tr^{n_\sigma}(\sigma)=\big[a,\textit{ret}=\textrm{false}\big]$,
no smaller value of $m_\sigma$ can satisfy the requirement that
$\Tr^{m_\sigma}\big(\func(\sigma)\big) =
  \func\big(\Tr^{n_\sigma}(\sigma)\big)$.
\end{example}

\paragraph{Checking if a {\shrinker} is simulation-inducing.} 
The \initiation{} and \safety{} requirements are simple to check. To facilitate checking the \consecution{} requirement, we do not allow arbitrarily large numbers $n_\sigma, m_{\sigma}$ but, rather, determine a bound $N$ on the value of $n_\sigma$ and a bound $M$ on the value of 
$m_{\sigma}$.
This makes the \consecution{} requirement stronger than required for soundness, but avoids the need to reason about
pairs of states that are reachable by traces
of unbounded lengths ($n_\sigma$ and $m_\sigma$).

\begin{figure}[t]
\centering
\begin{tikzpicture}[>=latex]
  \def\sza{8mm}
  \matrix at (0,0) [matrix of nodes,column sep=12mm,anchor=west] {
  	|(s0)| $\sigma_0$ & |(s1)| $\sigma_1$ &
  	|(s2)| $\sigma_2$ & |(s3)| $\cdots$ &
  	|(s4)| $\sigma_k$ & |(s5)| $\cdots$ \\
  };
  { [every path/.style={->,dashed}, 
     every node/.style={above,font=\smallish}]
  \draw (s0) -- node {$\Tr^{n_{\sigma_0}}$} (s1);
  \draw (s1) -- node {$\Tr^{n_{\sigma_1}}$} (s2);
  \draw (s2) -- node {$\Tr^{n_{\sigma_2}}$} (s3);
  \draw (s3) -- node {$\Tr^{n_{\sigma_{k-1}}}$} (s4);
  \draw (s4) -- node {$\Tr^{n_{\sigma_k}}$} (s5);
  }
  
  \node[right=0 of s5](safe) {$\models\varphi$};
  
  \matrix at (0,-1.25) [matrix of nodes,column sep=12mm,anchor=west] {
  	|(s'0)| $\sigma'_0$ & |(s'1)| $\sigma'_1$ &
  	|(s'2)| $\sigma'_2$ & |(s'3)| $\cdots$ &
  	|(s'4)| $\sigma'_j$ & |(s'5)| $\cdots$ \\
  };
  \draw[->] (s0) -- node[left] {\small$\func$} (s'0);
  \draw[->] (s1) -- node[left] {\small$\func$} (s'1);
  \draw[->] (s2) -- node[left] {\small$\func$} (s'2);
  \draw[->] (s4) -- node[left] {\small$\func$} (s'4);
  
  { [every path/.style={->,dashed},
     every node/.style={above,font=\smallish}]
  \draw (s'0) -- node {$\Tr^{m_{\sigma_0}}$} (s'1);
  \draw (s'1) -- node {$\Tr^{m_{\sigma_1}}$} (s'2);
  \draw (s'2) -- node {$\Tr^{m_{\sigma_2}}$} (s'3);
  \draw (s'3) -- node {$\Tr^{m_{\sigma_{j-1}}}$} (s'4);
  \draw (s'4) -- node {$\Tr^{m_{\sigma_j}}$} (s'5);
  }  

  \node[right=0 of s'5](safe') {$\models\varphi$};

  \coordinate(midlevel) at (0, -2.5);

  \matrix at (0,-3.8) [matrix of nodes,column sep=10mm,anchor=west] {
  	|(s+0)| $\sigma^{\dagger}_0$ & |(s+1)| $\sigma^{\dagger}_1$ &
  	|(s+2)| $\sigma^{\dagger}_2$ & |(s+3)| $\cdots$ &
  	|(s+4)| $\sigma^{\dagger}_r$ & |(s+5)| $\cdots$ \\
  };
  { [every path/.style={->,dashed}, 
     every node/.style={above,font=\smallish}]
  \draw (s+0) -- node {$\Tr^{m^{\dagger}_0}$} (s+1);
  \draw (s+1) -- node {$\Tr^{m^{\dagger}_1}$} (s+2);
  \draw (s+2) -- node {$\Tr^{m^{\dagger}_2}$} (s+3);
  \draw (s+3) -- node {$\Tr^{m^{\dagger}_{r-1}}$} (s+4);
  \draw (s+4) -- node {$\Tr^{m^{\dagger}_r}$} (s+5);
  }
  
  { [every node/.style={outer sep=1.25mm,circle}]
  \node(mid0) at (s'0 |- midlevel) { ~~ };
  \node(mid1) at (s'1 |- midlevel) { ~~ };
  \node(mid2) at (s'2 |- midlevel) { ~~ };
  \node(mid4) at (s+4 |- midlevel) { ~~ };
  }
  
  { [every node/.style={yshift=3pt}]   
  \node at (mid0) { $\vdots$ };
  \node at (mid1) { $\vdots$ };
  \node at (mid2) { $\vdots$ };
  \node at (mid4) { $\vdots$ };
  }
  
  { [every path/.style={dashed,decorate,decoration=
      {snake,amplitude=.5mm,segment length=8mm}},
     every node/.style={left}]
  \draw (s'0) -- node {\small$\func$} (mid0); 
     \draw[->] (mid0) -- node {\small$\func$} (s+0);
  \draw (s'1) -- node {\small$\func$} (mid1);
     \draw[->] (mid1) -- node {\small$\func$} (s+1);
  \draw (s'2) -- node {\small$\func$} (mid2);
     \draw[->] (mid2) -- node {\small$\func$} (s+2);
  \draw (s'4.south -| mid4) -- node {\small$\func$} (mid4);
     \draw[->] (mid4) -- node {\small$\func$} (s+4);
  }
  
  \node[below=0 of s+0.south east,anchor=west,rotate=-35,
        inner sep=0,outer sep=-1mm] {\tiny ${\in}$\smallish$\hspace{0.5pt}\Sigma_B$};

  \node[right=0.5 of s+5](safe+) {$\models\varphi$};
  
  { [every path/.style={shorten >=1mm, shorten <=1mm,draw=black!50!white}]
  \draw[->,dashed,decorate,decoration=
      {snake,amplitude=.5mm,segment length=7.5mm,pre length=2mm,post length=2.5mm}] (safe+) -- node[right] {$\Uparrow$} (safe');
  \draw[->] (safe') -- node[right] {$\Uparrow$} (safe);
  }
\end{tikzpicture}
\vspace{-3mm}
\caption{Soundness proof sketch; an arbitrary trace can be reduced to
  a low-\ranked{} trace by countable applications of $\func$.
  Since \rank{}s form a well-founded set, a base element is encountered
  after finitely many such reductions.
  Arrows with vertical ellipses indicate alternating applications of
  $\func$ and $\Tr^*$, except for initial states where
  \Cref{def:simulation}(1) ensures straight applications of $\func$ alone.
}
\label{shrinker:proof-sketch}
\end{figure}
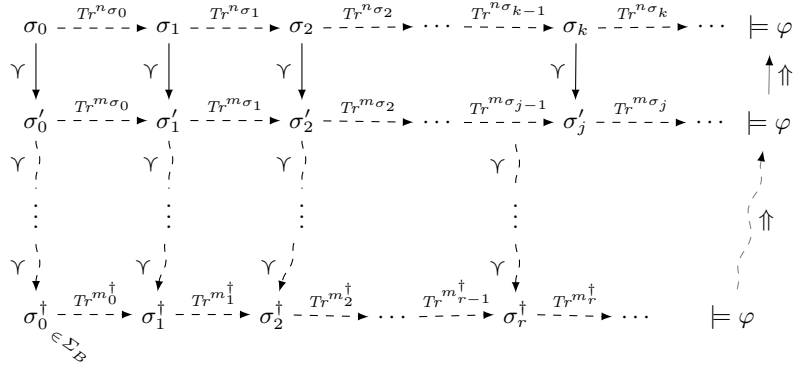

\paragraph{Using simulation-inducing {\shrinker} for safety verification.}

Roughly, the existence of a simulation-inducing {\shrinker} ensures that any counterexample to safety,
i.e., an execution starting from an initial state and ending in a \emph{bad}  state (a state that falsifies the safety property),
can be ``squeezed'' into a counterexample that starts from a ``smaller'' initial state. In this sense, the {\shrinker} establishes the induction step for proving safety by induction over the state \noamKeepEmph{rank}. 
To ensure the correctness of this argument, we need to require that a ``bad'' state
may not be ``skipped'' by the simulation induced by the {\shrinker}.

Formally, this is captured by the following definition.

\begin{definition}\label{def:size-preserving}\label{def:size-}
A transition system $\TS = (\setOfStates, \Init, \Tr, \Spec)$ is  \emph{\recidivist} if no ``bad'' state is a dead-end, i.e., $\sigma \not\in \Spec \implies \exists \sigma'.\,(\sigma,\sigma')\in \Tr$,  and that  transitions leaving ``bad'' states lead  to ``bad'' states, i.e., $\sigma \not\in \Spec \land (\sigma,\sigma')\in \Tr \implies \sigma' \not\in\Spec$.
\end{definition}
\Recidivism{} can be obtained by removing any outgoing transition of a bad state and adding a self loop instead. Importantly, this transformation does not affect the safety of the underlying program.
In our examples, terminal states of the program are treated as self loops, thus ensuring \recidivism.

\begin{lemma}\label{lem:reduce-cex}
Let $\func: \setOfStates \rightarrow \setOfStates$ be a simulation-inducing {\shrinker} 
for a
\recidivist{} transition system $\TS = (\setOfStates, \Init, \Tr, \Spec)$.
For every $\sigma_0 \in \setOfStates$, 
if there exists a counterexample that starts from $\sigma_0$, then there also exists a counterexample that starts from $\func(\sigma_0)$.
\end{lemma}
The proof is constructive: given a counterexample trace from $\sigma_0$, we use the simulation-inducing parameters $n_\sigma$ of the states $\sigma$ along the trace to divide it into segments such that the first and last state of each segment are the ones used as synchronization points for the simulation and the inner ones are the ones ``skipped'' over. We then match each segment $(\sigma,\sigma')$ with the corresponding trace of length $m_\sigma$ from $\func(\sigma)$ to $\func(\sigma')$, whose existence is guaranteed by the \consecution{} requirement. The concatenation of these traces forms a counterexample trace from $\func(\sigma_0)$. Formally:

\begin{proof}
Let $\tau = \sigma_0,\sigma_1,\ldots, \sigma_n$ be a counterexample trace starting from an initial state $\sigma_0 \in \Init$. 
If the counterexample is of length $0$, then $\func(\sigma_0)$ is also a counterexample of length 0 (by the \initiation{} and \safety{} requirements). Consider a counterexample of length $n > 0$.
We show how to construct a corresponding counterexample from $\func(\sigma_0)$.
We first split the indices $0,\ldots,n$ into (overlapping) intervals $I_0, \ldots, I_k$, where $I_0 = 0,\ldots,{n_\sigma}$, and for every $i \geq 1$, if the last index in $I_{i-1}$ is $j$ for $j <n$, then $I_i= j,\ldots,{j + n_{\sigma_j}}$. If $j + n_{\sigma_j} \geq n$, then $k:=i$. Since $\TS$ is \recidivist, we may assume, without loss of generality, that $j + n_{\sigma_j} = n$ (otherwise, because $\TS$ is \recidivist{} and $\sigma_n \not\in \Spec$, we can exploit one of the  transitions leaving $\sigma_n$, which necessarily exists and leads to a bad state, to extend the counterexample trace as needed.)
We denote by $\textit{first}(I_i)$, respectively $\textit{last}(I_i)$, the smallest, respectively largest, index in $I_i$.
By the definition of the intervals, for every $0 \leq i \leq k$, we have that $\textit{last}(I_i) = \textit{first}(I_i) + n_{\sigma_{{first}(I_i)}}$.
Hence, the \consecution{} requirement for $\sigma_{\textit{first}(I_{i})}$ ensures that there exists a trace of 
$m_{\sigma_{\textit{first}(I_{i})}}$ steps from $\func(\sigma_{\textit{first}(I_i)})$ to $\func(\sigma_{\textit{last}(I_{i})})$. Since $\sigma_{\textit{first}(I_{0})} = \sigma_0$ and for every $0<i\leq k$, $\sigma_{\textit{first}(I_{i})} = \sigma_{\textit{last}(I_{i-1})}$, we can glue these traces together to obtain a trace from $\func(\sigma_0)$ to $\func(\sigma_{\textit{last}(I_i)})$.
Finally, it remains to show that $\func(\sigma_{\textit{last}(I_{k})}) \not\in \Spec$. This follows from the \safety{} requirement, since ${\textit{last}(I_{k})} = n$, hence $\sigma_{\textit{last}(I_{k})} = \sigma_n \not\in \Spec$. 
\qed
\end{proof}

Ultimately, the existence of a simulation-inducing {\shrinker} implies that a counterexample can be ``squeezed'' to one that starts from a base initial state. Hence, to establish that the transition system is safe, it suffices to check that it is safe when the initial states are restricted to the base states, i.e., to $\Init \cap \setOfStates_B$.

\begin{theorem}[Soundness]
\label{thm:soundness}
Let $\func: \setOfStates \rightarrow \setOfStates$ be a simulation-inducing {\shrinker} with base $B$ for a 
\recidivist{} transition system $\TS = (\setOfStates, \Init, \Tr, \Spec)$.
If $\TS_B = (\setOfStates, \Init \cap \setOfStates_B, \Tr, \Spec)$ is safe then $\TS$ is safe.
\end{theorem}

\begin{proof}
Suppose for the sake of contradiction, that $\{\sigma_{i}\}_{i=0}^{d}$ is a counterexample trace with minimal \measure{} for $\sigma_0$ (such a state with a minimal \measure{} exists since $\preceq$ is well-founded). Since $\TS_B$ is safe, it must be that $\sigma_0 \in \setOfStates_{\overline{B}}$ (since $\sigma_0 \in \Init$,
while safety of $\TS_B$ ensures that no counterexample trace can start from $\Init \cap \setOfStates_B$).
By \Cref{lem:reduce-cex}, 
we have that $\func(\sigma_0)$ also has an outgoing counterexample trace.
However, since $\sigma_0 \in \setOfStates_{\overline{B}}$, we get that
$\m(\func(\sigma_0))\prec \m(\sigma_0)$, in contradiction to the minimality of $\sigma_0$.
\qed
\end{proof}

In all of our examples, the transitions of $\TS$ do not increase the \measure{} of the state. In such cases, we can also restrict the state space of $\TS_B$ (and accordingly $\Tr$) to the base states in $\setOfStates_B$. Furthermore, in these examples, the size of the state (array) also determines the length of the executions up to a terminal state. Hence, bounded model checking suffices to determine (unbounded) safety of $\TS_B$, and together with $\func$, also of $\TS$.

\begin{remark}
As evident from the proof of \Cref{thm:soundness}, it suffices to require that $\func$ decreases the {\rank} of the \emph{initial} non-base states, and not of all the non-base states.
\end{remark}

\section{Synthesizing Squeezing Functions}\label{Se:Syn}

\bgroup  
\newcommand\combo[2]{#1$_{\textit{#2}}$}
\newcommand\expr[1]{\combo{expr}{#1}}
\newcommand\var[1]{\combo{var}{#1}}
\newcommand\elem[1]{\combo{elem}{#1}}
\newcommand\const[1]{\combo{const}{#1}}

So far we have assumed that the squeezer $\func$ is readily available,
in much the same way that loop invariants are available ---
typically, as user annotations --- in standard unbounded loop verification.
As demonstrated by the examples in \Cref{Se:Overview,Se:shrinker},
$\func$ is specific to a given program and safety property.
Thus, it might be tedious to provide a different squeezer every time we wish to check a different safety property.
In this section we show how to lighten the burden on the user by automating the process of obtaining squeezing functions for a class of typical programs that loop
over arrays.

The solution for the \emph{squeezer-inference} problem we take in this paper is to utilize a rather standard
enumerative synthesis technique of multi-phase generate-and-test: We
take advantage of the relative simplicity of $\sq$
and provide a synthesis loop where we
generate grammatically-correct squeezing functions and  test whether they induce simulation.

\subsection{Generate}\label{Se:Syn:Generate}
First we note that while $\func$ is applied to arbitrary states in
\Cref{def:simulation}, it is only required to reduce the rank of
non-base states $\sigma\in B$.
For states $\sigma\in B$ it is trivial to satisfy all the requirements
by defining $\func(\sigma)=\sigma$.
In the sequel, we therefore only consider \shrinking{} functions whose
restriction to $B$ is the identity, and synthesize code for \shrinking{}
non-base states.

A central insight 
is that \shrinking{} functions $\func$ for different programs
still have some structure in common:
for programs with arrays, squeezing amounts to removing an element from
the array, and adjusting the index variables accordingly.
Some more detailed treatment may be needed for general purpose variables,
such as the accumulators \lstinline|l| and \lstinline|r| of
\lstinline|sum_bidi| (recall \Cref{intro:sum-bidi-loopinv}),
but the resulting expressions are still small.

\begin{figure}[t]
\newcommand\cmpop{$\diamond$\xspace}
\renewcommand\bar{$\,|\,$}
\renewcommand\arraystretch{1.2}
\newcommand\medium{\fontsize{8}{8}\selectfont}
\begin{tabular}{l@{~~}r@{~~}l}
body    & ::= & \lstinline|if (| cond \lstinline|)| \\[-1pt]
        &     & \qquad \lstinline|remove(| arr, \expr{index} \lstinline|)|
                $\big[$var$_{\textit{int}}$
                                   \lstinline|=| \expr{int}$\big]^*$\\
        &      & \lstinline|else|   \\[-3pt]
        &      & \qquad\lstinline|remove(| arr, \expr{index} \lstinline|)|
                 $\big[$var$_{\textit{int}}$
                                   \lstinline|=| \expr{int}$\big]^*$ \\[1pt]
cond    & ::= & \elem{$\tau$} \cmpop (\elem{$\tau$}\bar\const{$\tau$}) \\
           & \bar & \var{index} \cmpop (\var{index}\bar\const{index})
                    \hspace{1.5cm}
                    \cmpop ~~::=~~ \lstinline|==| {\bar} \lstinline|!=|
                            {\bar} \lstinline|<=| {\bar} \lstinline|>=| \\
           & \bar & cond \lstinline|&&| cond
                    \,~{\bar}~ cond \lstinline^||^ cond    \\
\expr{index}
        & ::= & \const{index} {\bar} \var{index}
                {\bar} \lstinline|len(|arr\lstinline|) - |%
                (\const{index}\bar\var{index}) \\
\expr{int}
        & ::= & \var{int} (\lstinline|+|{\bar}\lstinline|-|) \elem{int} \\
\elem{$\tau$}
        & ::= & arr \lstinline|[| \expr{index} \lstinline|]| \\
\const{index}
        & ::= & \lstinline|0| {\bar} \lstinline|1| {\bar} \lstinline|2|
\\
\const{$\tau$}
        & ::= & \lstinline|0| {\bar}
        {\medium other constants occurring in the program} \\
\multicolumn{3}{l}{arr, \var{index}, \var{int}, \var{char} ---
  {\medium identifiers occurring in the program}}
\end{tabular}
\caption{Program space for syntax-guided synthesis of~$\func$.
Expressions are split into three categories:
\textit{index}, \textit{int}, and \textit{char} as described in
\Cref{Se:Syn}.
$\tau\in\{\textit{int},\textit{char}\}$.}
\label{synth:grammar}
\end{figure}

We have found that, for \noamKeepEmph{the set of programs 
used in our experiments},
$\func$ can be characterised by the grammar in \Cref{synth:grammar}.
The grammar allows for functions comprised of a single if statement, where in each branch
an array is squeezed using the \lstinline|remove| function, and several integer variables are set.
Conditions are generating by composing array elements, local variables
and a fixed set of constants based on the given program,
with standard comparison operators and boolean connectives.
The semantics of \lstinline|remove(|arr, position\lstinline|)| are such that
a single element is removed from the array at the specified position,
and all index variables are adjusted by decrementing them if they
are larger from the index of the element being removed.
This behavior is hard-coded and is specific to array-based loops.
Our experience has shown that a single conditional statement is
indeed sufficient to cover many different cases (see \Cref{Se:ImpExper}).

To bound the search space, expressions and conditions have bounded sizes (in terms of AST height)
imposed by the generator \noamKeepEmph{and the user selects the set of basic predicates from which the condition of the \prog{if} statement is constructed}. The resulting  space, however, is still often too large to be explored efficiently.
To reduce it, some type-directed pruning is carried out so that
only valid functions are passed to the checker.
Moreover, our synthesis procedure distinguishes between variables that are used
as indices to the array (\var{index}) and regular integer variables
(\var{int}), and does not mix between them.
We further assume that we can determine, from analyzing the program's
source code, which index variable is used with which array(s).
So when generating expressions of the form
arr\lstinline|[|~\textit{i}~\lstinline|]| \textit{etc.},
only relevant index variables are used.
Also, we note that generated squeezers preserve in bounds access by construction.

\subsection{Test}\label{Se:Syn:Test}
The \emph{test} step checks whether a candidate squeezer that is generated by the synthesizer satisfies the requirements
of \Cref{def:simulation}.
For the simulation-inducing requirement, we restrict $n_{\cstate}=1..2$ and
$m_{\cstate}=0..1$.
The step is divided into three phases.
In the first phase, candidates are checked against a bank of concrete
program states (both reachable and unreachable).
In the second phase, candidates are verified for a bounded array size,
but with no restrictions on the values of the elements.
Those that pass bounded verification enter the third phase where
full, unbounded verification is performed.

The second and third phases of the test step require the use of an SMT solver.
The second phase is useful since incorrect candidates may cause the solver
to diverge when queried for arbitrary array sizes.
Limiting the array size to a small number (we used 6) enables to rule out
these candidates in under a second.
To simplify the satisfiability checks, we found it beneficial to decompose the verification task.
To do so, we take advantage of the structure of the squeezer, and split each satisfiability query (that corresponds to one of the requirements in \Cref{def:simulation})
into two queries, where in each query we make a different assumption regarding the branch the squeezer function takes.
We note in this context that the capabilities of the underlying solver direct
(or limit in some sense) the expressive power of the squeezer.
In this aspect, it is also worth mentioning that
sequence theory support for element removal helped to define squeezers format.

\noamKeepEmph{For the simulation inducing check, we further exploit the property that for the kind of programs and squeezers we consider, the transitions of the program
usually do not change the truth value of the condition of the if statement in the definition of the squeezer.
Namely, if $\cstate$ makes a transition to $\cstate'$ then either both of them satisfy the condition or both of them falsify it; either way, their
definition of $\func$ follows the same branch.
This form of preservation can be checked automatically using additional queries.
When it holds, we can consider the same branch of the squeezer program in both the pre- and post-states, thus simplifying the query for checking simulation.
Similarly, we can opportunistically split the transition relation of the program into branches (e.g., one that executes an iteration of the loop and one that exits the loop).
In most cases, the same branch that was taken for $\cstate$ is also the one that needs to be taken from $\func(\cstate)$ to establish simulation. This leads to another simplification of the queries, which is sound
(i.e., never concludes that the simulation-inducing requirement holds when it does not), but potentially incomplete.
We can therefore use it as a ``cheaper'' check and resort to the full check if it fails.
}

\subsection{Filtering out unreachable states}
\label{Se:Unreachable}
For soundness, a squeezer needs to satisfy \Cref{def:simulation}
only on the reachable states. As we do not have a description of this set, for otherwise the verification task would be \noamKeepEmph{essentially}{} voided, we need to ensure that the requirements of simulation-inducement on a safe over-approximation of this set. A simple over-approximation would be the set of all states. However, this over-approximation might be too coarse, indeed we noticed in our experiments that in some cases, unreachable states have caused phases $1$, $2$ and $3$ to produce false negatives,\ie, disqualify squeezers which can  be used safely to verify the program.
Therefore we used an over-approximation of reachable states using
\begin{enumerate}
\item Bound constraints on the index variables: the index is expected to be within bounds of the traversed array. This property can be easily verified using other verifiers or by applying our verifier in stages, first proving this property and then proving the actual specification of the verified procedure under the assumption that the property hold.
\item 2-step bounded reachability:
We found out that for our examples, looking only at states that are
reachable from another state in at most two steps is a general enough inclusion criterion.
Note that we do not require 2-step reachability from an initial state, but rather from \emph{any} state, hence this set
over-approximates the set of reachable states.
\end{enumerate}

\egroup

\section{Implementation and Experimental Results}
\label{Se:ImpExper}
We implemented an automatic verifier for array programs based on our approach,
and applied it successfully to verify natural properties of a few \noamKeepEmph{interesting}{} 
array-manipulating programs.

\paragraph{Base case.}
We discharged the base case of the induction (the verification on the base states)
using KLEE~\cite{KLEE}---a state-of-the-art symbolic execution~\cite{Cadar:2013} engine.
It took KLEE less than one tenth of  a second to verify the correctness of each program in our benchmarks on the states in its base.
This part of our verification approach is standard, and we discuss it no further;
in the rest of this section we focus on the generation of the squeezing functions.

\subsection{Implementation}
\label{Se:Imp}


The generate step and phase $1$ of the test step of
the squeezer synthesizer were implemented using a standalone C++ application
 that generates all
$\func$ candidates with an AST of depth three.
Each squeezer was tested on a pre-prepared state bank and every time a squeezer passed the tests it was immediately passed on to phase 2.
The state bank contained states with arrays of length five or less.
For each benchmark,
we used
up to 24,386 states with randomly selected array contents.
The number of states was determined as follows:
Suppose the program state is comprised of
$k$  variables and an array of size $n$.
We randomly selected $p$ elements that can populate the array:
$p=\{'a','b',0\}$ for string manipulating procedures and
$p=\{-4,-2,9,100,200 \}$ for programs that manipulate integer arrays.
We determined the number test states according to the following formula:
$d^k \cdot |p|^n/df$, where $df$ is an arbitrary dilution factor used to reduce the number of states from thousands to hundreds. (In our experiments, $df=17$.)

The second and third phases  were implemented using Z3~\cite{Z3},
a state of the art SMT solver.
We chose to use the theory of sequences,
since its API allows for a straightforward definition of the operation
{remove\lstinline|(|arr\lstinline|,|$i$\lstinline|)|} (see \Cref{synth:grammar}).
In practice, the sequence solver proved to be overall more effective than a
corresponding encoding using the more mature array solver.
In that aspect, it is worth mentioning that verifying fault preservation on its own 
\emph{is} faster with the theory of arrays. We conjecture that this is because the specification has quantifiers  while the other requirements can be verified using quantifier-free reasoning.

The transition relation was manually encoded in SMT-LIB2 format.
However, it should be straightforward to automate this step.


\subsection{Experimental Evaluation}
\label{Se:ExperRes}

We evaluated our technique by verifying a few array-manipulating  programs against
their expected specifications.
The experiments were executed on a laptop
with Intel i7-8565 CPU (4 cores) with 16GB of RAM running Ubuntu 18.04.

\paragraph{Benchmarks.}
We ran our experiments on seven  array-manipulating programs:
\prog{strnchr} 
looks for the first appearance of a given character in the first
$n$ characters of a string buffer.
\noamKeepEmph{\prog{strncmp} compares whether two strings are identical up to their first $n$ characters or the first zero character}.
\prog{max\_ind} (resp. \prog{min\_ind})
looks for the index of the maximal (resp. minimal) element in an integer array.
\prog{sum\_bidi} 
is our running example.
\noamKeepEmph{\prog{is\_sorted} 
checks if the elements of an array are sorted in an increasing order.}
\noamKeepEmph{\prog{long\_pref} 
is looking for the longest prefix of an array comprised of either a monotonically increasing or a monotonically decreasing sequence.}

The user supplies predicates that are used when synthesizing each squeezer.
These were selected based on understanding what the program does and the
operations it uses internally.
\textit{E.g.}, for \lstinline|strncmp| equality comparisons between
same-index elements of the two input arrays are used (\verb"s1[0]==s2[0]" etc.), as well as
comparison with constant 0;
for \lstinline|long_pref|, order comparisons (\verb"s1[1]<=s1[2]" etc.) between different elements
of the same array are used instead.




\begin{table}[t]
\begin{tabular}{l|c|r|r|r|r|r|r|r|rcl|r}
\multicolumn{3}{ c }{           } &
\multicolumn{3}{|c|}{\bf Phase 1} &
\multicolumn{2}{|c|}{\bf Phase 2} &
\multicolumn{1}{|c|}{\bf Phase 3} & 
\multicolumn{3}{|c|}{\bf Total Time} &
\multicolumn{1}{|c }{\QUIC      }
\\
\textbf{Program   } &
\textbf{B         } &
\textbf{\#\,Cand  } &
\textbf{\scriptsize $|$Bank$|$} &
\textbf{\scriptsize Test    } &
\textbf{\scriptsize Time      } &
\textbf{\scriptsize BMC     } &
\textbf{\scriptsize Time      } & 
\multicolumn{1}{|c|}{\bf\scriptsize Time} &
\multicolumn{3}{|c|}{\bf\scriptsize G\&T+KLEE} &
\multicolumn{1}{|c }{\bf\scriptsize Time}
\\ \hline
\prog{strnchr} &
2              & 
80             & 
356            & 
29             & 
0.004          & 
1              & 
0.12           & 
0.16           & 
0.28&+&0.07      & 
0.32             
\\
\prog{strncmp} &
2                & 
980              & 
76               & 
196              & 
0.02             & 
1                & 
7.2              & 
154.48           & 
161.70&+&0.05        & 
0.19               
\\
\prog{max\_ind} &
2               & 
8000            & 
368             & 
10              & 
0.18            & 
2               & 
1.86            & 
4.44            & 
4.73 &+& 0.05   & 
0.11              
\\
\prog{min\_ind} &
2               & 
8000            & 
257             & 
9               & 
0.26            & 
2               & 
2.1             & 
16.86           & 
17.21 &+& 0.05  & 
0.09              
\\
\prog{sum\_bidi} &
2                & 
6328125          & 
4602             & 
1200             & 
2.18             & 
1                & 
0.57             & 
0.61             & 
3.36 &+& 0.05    & 
t.o.               
\\
\prog{is\_sorted} &
4                & 
900              & 
25736            & 
764              & 
4.37             & 
1                & 
0.59             & 
0.67             & 
5.63 &+& 0.06    & 
0.15               
\\
\prog{long\_pref} &
3                & 
6480             & 
24386            & 
4696             & 
22.93            & 
1                & 
1.25             & 
0.89             & 
25.07&+&0.05     & 
t.o.               
\\
\end{tabular}
\vspace{0mm}
\caption{Experimental results (end-to-end). Time in seconds. \colbf{G\&T} is a shorthand for Generate\&Test}
\label{Table_Experimental_results}
\end{table}

\paragraph{Results.}
\Cref{Table_Experimental_results}
describes the end-to-end running times of our verifier, i.e.,
the time it took our tool to establish the correctness of each example.
\noamKeepEmph{In this experiment, every candidate squeezer was tested before the next squeezer was generated. }
The table shows the time it took the synthesizer to find the first simulation-inducing squeezer plus the time it took to establish the correctness of the programs on the states in the base using KLEE  (\colbf{Total Time}).
The table also compares our verifier to \QUIC~\cite{ATVA18}, an automatic synthesizer of loop invariants.
In general, when both tools where able to prove that the analyzed procedure is correct, {\QUIC} was somewhat faster, and in the case of \prog{strncmp} much faster. However, on two of our benchmarks
{\QUIC} timed out ($1$ hour) whereas our tool was able to prove them correct in less than 30 seconds.

\Cref{Table_Experimental_results}
also provides more detailed statistics regarding the experiments:
The  rank of the base states (\colbf{B}),
the total number of possible candidates based on the supplied predicates and the bound on the depth of the AST 
(\colbf{\#\,Cand}),{} and a more detailed view of each phase in the testing step.
For phase 1, it reports the number of states in the pre-prepared state bank (\colbf{$|$Bank$|$}),
the number of squeezers tested until a simulation-inducing one was found (\colbf{Test}),
and the total time spent
\noamKeepEmph{to test these squeezers (\colbf{Time})}.
For phase 2, it reports the number of candidates which passed phase 1 and survived bounded verification (\colbf{BMC})
and the time spent in this phase  (\colbf{Time}).
\noamKeepEmph{For phase 3, we  report how many simulation-inducing squeezers were found the time it took to apply full verification}.

In all our experiments except of \prog{max/min\_ind} only the simulation-inducing squeezers passed bounded verification. In the latter case, a squeezer passed BMC due to the use of arrays of size at most five
where the cells $a[2]$ and $a[n-2]$ are adjacent. Had we increased the array bound to six, these false positives would have been eliminated by the bounded verification.



\begin{table}[t]
\centering
\begin{tabular}{l|r|r|r|r|r|r|r|r|r|r}
\multicolumn{1}{ c }{           } &
\multicolumn{4}{|c|}{\bf Phase 1} &
\multicolumn{3}{|c|}{\bf Phase 2} &
\multicolumn{3}{|c}{\bf Phase 3}
\\
\textbf{Program}     &
$|$\textbf{\scriptsize Pos.}$|$  &
   \textbf{\scriptsize Time}     &
$|$\textbf{\scriptsize Neg.}$|$ &
   \textbf{\scriptsize Time}     &
$|$\textbf{\scriptsize Pos.}$|$  &
   \textbf{\scriptsize Time}     &
  $\textbf{\scriptsize Time}_{| \textbf{\tiny Neg.} |}$ &
$|$\textbf{\scriptsize Pos.}$|$  &
   \textbf{\scriptsize Time}     &
  $\textbf{\scriptsize Time}_{| \textbf{\tiny Neg.} |}$
\\ \hline
\prog{strnchr} &
1              & 
$\epsilon$     & 
9              & 
$\epsilon$     & 
1              & 
0.94           & 
$-$            & 
1              & 
0.98           & 
$-$              
\\
%
\prog{strncmp} &
3              & 
$\epsilon$     & 
36             & 
$\epsilon$     & 
3              & 
14.29          & 
$-$            & 
3              & 
154.48         & 
$-$              
\\
\prog{max\_ind} &
11              & 
$\epsilon$      & 
3               & 
$\epsilon$      & 
2               & 
0.78            & 
1.08            & 
1               & 
31.00              & 
0.41              

\\
\prog{min\_ind} &
11              & 
$\epsilon$      & 
7               & 
$\epsilon$      & 
2               & 
0.91            & 
1.19            & 
1               & 
16.00           & 
0.43              
\\
\prog{sum\_bidi} &
12               & 
$\epsilon$       & 
1                & 
0.05             & 
1                & 
0.56             & 
0.69             & 
1                & 
0.61             & 
$-$                

\\
\prog{is\_sorted} &
1                & 
$\epsilon$       & 
18               & 
$\epsilon$       & 
1                & 
0.59             & 
$-$              & 
1                & 
0.67             & 
$-$                

\\
\prog{long\_pref} &
2                 & 
$\epsilon$        & 
74                & 
$\epsilon$        & 
1                 & 
1.03              & 
1.22              & 
1                 & 
0.89              & 
$-$                 
\\
\end{tabular}
\vspace{2mm}
\caption{Experimental results. Time in seconds. $\epsilon \leq 0.0001$ \NRX{Check ver times for min + sum (both 16)}}
\label{Table_Experimental_results_statistics_data}
\end{table}

\Cref{Table_Experimental_results_statistics_data}
provides average times required to pass all the generated squeezers through the testing pipeline.
For phase 1, it reports the number of squeezers which passed (\colbf{Pos}) resp. failed (\colbf{Neg}) testing against the randomly generated states and the average time it took to test the squeezers in each category (\colbf{Time}).
The table reports the statistics pertaining to phase 2 and 3 in a similar manner, \noamKeepEmph{except that it omits the number
of squeezers which failed the phase as this number can be read off the number of squeezers which reached this phase.}


\begin{table}[t]
\centering
\newcommand\zno[1]{{\smallish(#1)}~}
\renewcommand\arraystretch{1.1}
\begin{tabular}{l@{~}|@{~}l}
\textbf{Program}  & \textbf{\Shrinker} \\ \hline
\prog{strchr(c)} &  \prog{if ( s[0] == c || s[0]==0 ) remove(s,1) else remove(s,0)}
\\
\prog{strncmp} & \zno1\lstinline|if (s1[0] == s2[0] && s1[0] != 0) remove(s1,0); remove(s2,0)| \\
               & ~~~~\lstinline|else remove(s1,1); remove(s2,1)|
\\
               & \zno2\lstinline|if (s1[0] == s2[0] && s2[0] != 0) remove(s1,0); remove(s2,0)| \\
               & ~~~~\lstinline|else remove(s1,1); remove(s2,1)|
\\
               & \zno3\lstinline^if (s1[0] != s2[0]) || (s1[0] == 0 && s2[0] == 0))^\\
               & ~~~~\lstinline|     remove(s1,1); remove(s2,1)| \\
               & ~~~~\lstinline|else remove(s1,0); remove(s2,0)|
\\
\prog{max\_ind} & \prog{if (s[n-2] <= s[n-1]) remove(s,n-2) else remove(s,n-1)}
\\
\prog{is\_sorted} & \prog{if (s[n-3]<=s[n-2]<=s[n-1]) remove(s,n-1)} \prog{else remove(s,n-4)}
\\
\prog{long_pref} & \prog{if ((s[0]<=s[1]<= s[2]) || (s[0]>s[1]>s[2]))} \prog{remove(s,0) } \\
                 & \prog{else remove(s,n-1)}
\end{tabular}
\vspace{0mm}
\caption{Syntesized squeezers. \noamKeepEmph{$n$ is the size of the input array}}
\label{Ta:ExperFunc}
\end{table}

\Cref{Ta:ExperFunc} shows some of the automatically generated squeezers. \noamKeepEmph{We
obtained a single simulation-inducing squeezer in all of our tests except for \prog{strncmp} where three squeezers were synthesized. The three differ only syntactically by the condition of the \prog{if} statements. However, semantically, the three conditions are equivalent. Thus, improving the symmetry-detection optimizations to include equivalence up-to-de morgan rules would have filtered out two of the three squeezers.}

\section{Related Work}
\label{Se:Related}

\begin{paragraph}{}
Automatic verification of
infinite-state systems, \ie, systems where the size of an individual
state is unbounded
such as numerical programs (where data is considered unbounded),
array manipulating programs (where both the length of the array and the data it contains may be unbounded),
programs with dynamic memory allocation (with unbounded number of dynamically-allocatable memory objects),
and parameterized systems (where, in most cases, there is an unbounded number of instances of finite subsystems) is a long standing challenge in the realm of formal methods.
\end{paragraph}

\begin{paragraph}{Well structured transition systems.}
Well structured transition systems (WSTS)~\cite{DBLP:journals/tcs/FinkelS01,DBLP:journals/iandc/AbdullaCJT00,DBLP:conf/lics/AbdullaCJT96} are a class of infinite-state transition systems for which safety verification is decidable, with a backward reachability analysis being a decision procedure.
In these transitions systems, the set of states is accompanied by a well-quasi order that induces a simulation relation: a state is simulated by those that are ``larger'' than it.
As a result, the set of backward-reachable states is upward closed. The simulation-inducing well-quasi order used in WSTS
resembles our condition of a simulation-inducing squeezer. However, there are several fundamental differences:
\begin{inparaenum}[(i)]
\item The order underlying our technique is required to be well-founded, which is a strictly weaker requirement than that of a well-quasi order; 
\item The simulation-inducing requirement requires each state to be simulated by its squeezed version, which has a \emph{lower} rank rather than greater;
further, a state need not be simulated by \emph{every} state with a lower rank; accordingly, the set of backward-reachable states need not be upward (nor downward) closed.
\item Our procedure is not based on backward (or any other form of) reachability analysis.
\end{inparaenum}
\end{paragraph}

\begin{paragraph}{Reductions.}
Cutoff-based techniques, e.g.,~\cite{Emerson00}, reduce model checking of unbounded parameterized systems to model checking for systems of size (up to) a small predetermined cutoff size.
Verification based on dynamic cut-offs~\cite{CutOff,Abdulla13} also considers parameterized systems but employs a verification procedure which can dynamically detect cut-off points beyond which the search of the state space need not continue.
Invisible invariants~\cite{InvisibleNeither,InvInv} are used to verify unbounded parameterized systems in a bounded way.
The idea is to use the standard deductive invariance rule for proving invariance properties but consider only bounded systems for discharging the verification conditions, while ensuring that they hold for the unbounded system. The approach provides
(i) a heuristic to generate a candidate inductive invariant for the proof rule, and (ii) a method to validate the premises of the proof rule once a candidate is generated~\cite{InvisibleNeither}.
\end{paragraph}

Similar reductions were applied to array programs--a particular form of parameterized systems but with unbounded data--as we consider in this work.
For example, in~\cite{DBLP:conf/tacas/KumarSVS18-rev2}, \emph{shrinkable} loops are identified as
loops that traverse large or unbounded arrays but may be soundly replaced by a bounded number of nondeterministically chosen iterations; and
in~\cite{DBLP:conf/sas/MonniauxG16-rev5}, abstraction is used to
replace reasoning about unbounded arrays and quantified properties by reasoning about a bounded number of array cells.

A fundamental difference between our approach and these works is that we do not reduce the problem to a bounded verification problem. Instead, we generate verification conditions which amount to a proof by induction on the size of the system.
\noamKeepEmph{In fact, from the perspective of deductive verification, our work can be seen as introducing a new induction scheme.}

\begin{paragraph}{Loop invariant inference.}
Arguably, inference of loop invariants is the ubiquitous  approach for automatic verification of
infinite-size systems. 
Recent research efforts in the area have concentrated around inference of quantified invariants,
in particular, the search for universal loop invariants is a central issue.

Classical predicate
abstraction~\cite{DBLP:conf/cav/GrafS97,DBLP:conf/tacas/BallPR01} has
been adapted to quantified invariants by extending predicates with
\emph{skolem} (fresh)
variables~\cite{DBLP:conf/popl/FlanaganQ02,DBLP:conf/vmcai/LahiriB04}.
This is sufficient for discovering complex loop invariants of array
manipulating programs similar to the simpler programs used in our
experiments.

A research avenue that has received ongoing popularity is the use of
constrained Horn clauses (CHCs) to model properties of transition
systems which have been used for inference of universally quantified
invariants~\cite{DBLP:conf/sas/BjornerMR13,DBLP:conf/sas/MonniauxG16,fse16}
by limiting the quantifier nesting in the loop invariant
being sought.
In~\cite{DBLP:conf/cav/FedyukovichPMG19-rev4}, universally quantified solutions (inductive invariants) to CHCs
are inferred via syntax-guided synthesis.

Another active research area is Model-Checking Modulo Theories
(MCMT)~\cite{DBLP:conf/cade/GhilardiR10} which extends model checking to
array manipulating programs and has been used for verifying
heap manipulating programs and parameterized systems
(e.g.,~\cite{DBLP:conf/fmcad/ConchonGKMZ13}) using
quantifier elimination techniques.
For example, in \textsc{Safari}~\cite{DBLP:conf/cav/AlbertiBGRS12} (and later
\textsc{Booster}~\cite{DBLP:conf/atva/AlbertiGS14}), the theory of arrays~\cite{DBLP:journals/corr/abs-1204-2386} is used to construct a QF proof  of bounded safety which is generalized by universally quantifying out some terms.

\textsc{IC3}~\cite{DBLP:conf/vmcai/Bradley11} extends predicate abstraction into a framework in which the predicate
discovery is directed by the verification goal and heuristics are used to generalize proofs of bounded depth execution to inductive invariants.
\textsc{UPDR}~\cite{DBLP:conf/cav/KarbyshevBIRS15} and \textsc{Quic3}~\cite{ATVA18}
extend  \textsc{IC3} to quantified invariants.
UPDR focuses on programs specified 
using the Effectively PRopositional (EPR) fragment of \emph{uninterpreted} first order logic (e.g., without arithmetic) for
which quantified satisfiability is decidable. As such, UPDR does not
deal with quantifier instantiation.
\textsc{Quic3}  uses model based projection
and generalizations
based on bounded exploration.

Like these techniques we also use heuristics to overcome the unavoidable undecidability barrier.
In our case, this amounts to the selection of the squeezing function.
In contrast to all the aforementioned approaches, our technique does not rely on the inference of loop invariant but rather proves programs correct by induction on the size (rank) of their states.
\end{paragraph}

\begin{paragraph}{}
We note that we do not position our technique as a replacement to automatic inference of loop invariants but rather as a complementary approach. Indeed, while some tricky properties can be easily verified by our approach, e.g.,
the postcondition of \lstinline|sum_bidi|, a property which we believe no other automatic technique can deduce, other properties which are simple to establish using loop invariants, e.g., that variable \lstinline|i| is always in the range $0..n\!-\!1$, are surprisingly challenging for our technique to establish.
\end{paragraph}

\begin{paragraph}{Recurrences.}
Other approaches represent the behavior of loops in array-programs via recurrences defined over an explicit loop counter,
and use these recurrences to directly verify post-conditions with universal quantification over the array indices.
In~\cite{DBLP:conf/vstte/RajkhowaL18-rev1} this is done by customized instantiation schemes and explicit induction when necessary.
In~\cite{DBLP:conf/sas/ChakrabortyGU17-rev3}, verification is done by identifying a relation between loop iterations (characterized by the loop counter)
and the array indices that are affected by them, and verifying that the post-condition holds for these indices.
Similarly to our approach, these works do not rely on loop invariants, but they do not allow to verify global properties over the arrays, such as
the postcondition of \lstinline|sum_bidi|.
\end{paragraph}

\begin{paragraph}{Program synthesis.}
The inference we use for $\func$ is indeed a form of program synthesis,
as was alluded to in \Cref{Se:Overview} by representing $\func$ via pseudo-code.
In particular, \emph{syntax-guided synthesis} (SyGuS) \cite{alur2015syntax} is the domain of program synthesis where the target program is derived from a programming language according to its syntax rules.
\cite{itzhaky2016deriving,udupa2013transit,farzan2019modular,wang2017synthesizing} all fall within this scope.

\emph{Sketching} is a common feature of SyGuS.
The term is inspired by Sketch~\cite{solar2006combinatorial}, referring to the
practice of giving synthesizers a program skeleton with a missing piece or pieces.
This uses domain knowledge to reduce the size of the candidate space.
It is quite common to use a domain-specific language (DSL) for this purpose%
~\cite{srivastava2010program,smith2016mapreduce,srivastava2013template,hua2017edsketch,wang2018solver}.
\cite{osera2015type} restricts programs by typing rules in addition to just syntax.
\cite{gulwani2016pbe} develops it further by restricting how operators may be composed.
Our synthesis procedure (\Cref{Se:Syn}) follows the same guidelines:
the domain of array-scanning programs dictates the constructed space of
squeezer functions, and moreover, inspecting the analyzed program allows
for more pruning by (i) matching index variables to array variables and
(ii) focusing on operators and literal values occurring in the program.
This early pruning is responsible for the feasibility of our synthesis
procedure, which apart from that is rather naive and does not facilitate
clever optimizations such as equivalence reduction%
~\cite{osera2015type,feser2015synthesizing}.
\end{paragraph}

\section{Conclusions}\label{Se:Conc}

At the current state of affairs in automatic software verification of infinite state systems,
the scene is dominated by various approaches with a common aim:
computing over-approximations of unbounded executions by means of inferring loop invariants.
Indeed, \emph{abstract interpretation}~\cite{CousotCousot77},
\emph{property-directed reachability}~\cite{DBLP:conf/vmcai/Bradley11},
unbounded model checking~\cite{Lazy06},
or
template-based verification~\cite{Srivastava2013}
can be seen as different techniques for computing such approximations by finding inductive loop invariants
which are tight enough not to intersect with the set of bad behaviors.
Experience has shown that
these invariants are frequently quite hard to come by, even for seemingly
simple and innocuous program, both automatically and manually.
The purpose of this paper is to suggest an alternative kind of correctness
witness, which \noamKeepEmph{may be}{} more amenable to automated search.
We successfully applied our novel verification technique to array programs and managed to prove programs and properties which are beyond the ability of existing automatic verifiers.
We believe that our approach can  be combined with standard techniques to give rise to a new kind of hybrid
techniques, where, e.g., a \emph{partial} loop invariant is used as a baseline ---
verified via standard techniques ---
and is then \emph{strengthened} to the desired safety property via
{\shrinker}-based verification.

{\mediumish
\paragraph{Acknowledgements.}
The research leading to these results has received funding from the European Research Council under the European Union's Horizon 2020 research and innovation programme (grant agreement No [759102-SVIS]), the Lev Blavatnik and the Blavatnik Family foundation, Blavatnik Interdisciplinary Cyber Research Center
at Tel Aviv University, Pazy Foundation, Israel
Science Foundation (ISF) grants No. 1996/18 and 1810/18,
and the Binational Science Foundation (NSF-BSF) grant 2018675.
}

\newpage

\bibliographystyle{splncs04}
\bibliography{biblio,srefs,bib-hila}

\end{document}